\documentclass[letterpaper, 10 pt, conference]{ieeeconf}  % Comment this line out
\IEEEoverridecommandlockouts
\usepackage{color}
\usepackage{caption}
\usepackage{enumerate}
\usepackage{cite}
\usepackage{empheq}
\usepackage{mathrsfs}

\usepackage{mathtools}

\usepackage{tikz}
\usepackage{circuitikz}
\usepackage{balance}

\usepackage[font=small,labelfont=bf]{caption}

\usepackage{subfig}

\makeatletter
\pgfcircdeclarebipole{}{\ctikzvalof{bipoles/interr/height 2}}{spst}{\ctikzvalof{bipoles/interr/height}}{\ctikzvalof{bipoles/interr/width}}{

    \pgfsetlinewidth{\pgfkeysvalueof{/tikz/circuitikz/bipoles/thickness}\pgfstartlinewidth}

    \pgfpathmoveto{\pgfpoint{\pgf@circ@res@left}{0pt}}
    \pgfpathlineto{\pgfpoint{.6\pgf@circ@res@right}{\pgf@circ@res@up}}
    \pgfusepath{draw}   
}

\def\pgf@circ@spst@path#1{\pgf@circ@bipole@path{spst}{#1}}
\tikzset{switch/.style = {\circuitikzbasekey, /tikz/to path=\pgf@circ@spst@path, l=#1}}
\tikzset{spst/.style = {switch = #1}}
\makeatother

\makeatletter
\let\proof\@undefined                        % undefine \proof
\let\endproof\@undefined                  % undefine \endproof
\makeatother
\usepackage{graphicx,amssymb,amstext,amsmath,amsthm}

\usepackage[bookmarks=true]{hyperref}
\usepackage{algorithm,algorithmicx,algpseudocode}
\algnewcommand{\algorithmicgoto}{\textbf{go to}}%
\algnewcommand{\Goto}[1]{\algorithmicgoto~\ref{#1}}%
\algnewcommand{\LineComment}[1]{\Statex \(\triangleright\) #1}
\algnewcommand{\LineCommentN}[1]{\Statex \hspace{1cm}\(\triangleright\) #1}

\usepackage{multirow}
\usepackage{stfloats}

\newtheorem{prop}{Proposition} % this could go into the preamble

\newtheorem{thm}{Theorem}
	\newtheorem{assumption}{Assumption}
\newtheorem{lem}{Lemma}
\newtheorem{defn}{Definition}
\newtheorem{rem}{Remark}
\newtheorem{problem}{Problem}

\usepackage{setspace}
\let\oldbibliography\thebibliography
\renewcommand{\thebibliography}[1]{%
  \oldbibliography{#1}%
}

\newcommand{\tdiag}{\textstyle{\mathrm{diag}}}

\newcommand{\yong}[1]{{\color{black} #1}}

\newcommand{\fa}[1]{{\color{black} #1}}
\newcommand{\mk}[1]{{\color{black} #1}}

\usepackage{flushend}
\begin{document}

% paper title
\title{\LARGE \bf Optimal Dynamic Control of Bounded Jacobian Discrete-Time Systems via Interval Observers} % Matrices} }

% You will get a Paper-ID when submitting a pdf file to the conference system
%\author{Author Names Omitted for Anonymous Review. Paper-ID Sze Zheng Yong}
\author{%
Mohammad Khajenejad\\
\thanks{%$^1$ The authors are with the Laboratory for Information and Decision Systems,
%Massachusetts Institute of Technology, Cambridge, MA, USA (e-mail: szyong@mit.edu, mhzhu@mit.edu, frazzoli@mit.edu).
M. Khajenejad is with the Departments of Mechanical Engineering, University of Tulsa, Tulsa, OK, USA (e-mail: \texttt{mok7673@utulsa.edu}).}
%\thanks{This work is partially supported by National Science Foundation grants CNS-1932066 and CNS-1943545.}
%\vspace{-0.35cm}
}

\maketitle
\thispagestyle{empty}
\pagestyle{empty}

\begin{abstract}
This paper presents an optimal dynamic control framework for bounded Jacobian nonlinear discrete-time (DT) systems with nonlinear observations affected by both state and process noise. Rather than directly stabilizing the uncertain system, we focus on stabilizing an interval observer in a higher-dimensional space, whose states bound the true system states. Our nonlinear dynamic control method introduces added flexibility over traditional static and linear approaches, effectively compensating for system nonlinearities and enabling potentially tighter closed-loop intervals. Additionally, we establish a separation principle that allows for the design of observer and control gains. We further derive tractable matrix inequalities to ensure system stability in the closed-loop configuration. The simulation results show that the proposed dynamic control approach significantly outperforms a static counterpart method.
 \end{abstract}

%\vspace{-0.125cm}
\section{Introduction} %\vspace{-0.05cm}
%\emph{Motivation.} 
Estimation and control in real-world systems are often hindered by significant nonlinearities and model uncertainties. In areas like engineering settings, systems biology, neuroscience, or collaborative robotics, complex dynamics and interactions make accurate modeling and parameter identification challenging \cite{izhikevich2007dynamical}. Such uncertainties cannot be ignored and necessitate robust, adaptive design strategies. Yet, conventional control methods often rely on precise models—an impractical assumption in the presence of noise, modeling errors, or external stochastic noise/disturbances. This highlights the need for estimation-based control approaches that remain effective under bounded, possibly non-stochastic uncertainties where classical Kalman filter-based approaches may fall short.

In such scenarios, set-valued observers, particularly interval-valued observers~\cite{khajenejad2024distributed}, have been widely used to estimate a range of admissible values for system states at each time step, using input-output data and bounds on uncertainties~\cite{tahir2021synthesis,khajenejad2021intervalACC}. Techniques to minimize interval length include applying monotone systems theory~\cite{farina2000positive} to achieve cooperative observer error dynamics, using interval arithmetic or Müller’s theorem~\cite{kieffer2006guaranteed}, and designing observer gains via linear, semidefinite, or mixed-integer programming~\cite{efimov2013interval,9790824,9803272}, where gains scale with model uncertainty magnitudes.

For control design, interval observer-based methods have been explored, primarily in linear continuous-time (CT) systems. For example, the research in~\cite{wang2020intervalcontrol} introduced an interval observer-based state feedback controller for switched linear systems, and the work in~\cite{zhang2024semi} developed a semi-global interval observer framework for robust control in linear time-invariant (LTI) systems with input saturation.

In nonlinear settings, interval observer designs have been applied with some limitations. In the work in~\cite{efimov2013control}, a dynamic output feedback approach used an interval state observer for a class of nonlinear CT systems, though it was restricted to linear observations and neglected measurement noise. Similarly, the authors in~\cite{he2018control} proposed a Luenberger-like interval observer for nonlinear CT systems with linear observations, also excluding noise considerations. More recently,~\cite{efimov2024interval} introduced a nonlinear feedback control strategy for a class of nonlinear CT systems using interval estimates, with stability conditions for estimation and regulation errors derived via linear matrix inequalities (LMIs). However, these works did not address nonlinear discrete-time (DT) systems subject to noise in both state and measurements.

In our previous work~\cite{khajenejad2024optimalcontrolstatic}, we proposed a static feedback control strategy for bounded Jacobian nonlinear DT systems with nonlinear observations, incorporating both state and process noise. This paper aims to advance the approach in \cite{khajenejad2024optimalcontrolstatic} by introducing the following improvements.

\emph{{Contributions}.} We address the design of an optimal dynamic control strategy for stabilization and disturbance attenuation in bounded Jacobian nonlinear discrete-time (DT) systems with uncertainty. This strategy is founded on the dynamic stabilization of a higher-dimensional interval observer, which tightly bounds the true system states from above and below. Unlike static approaches~\cite{efimov2013control,khajenejad2024optimalcontrolstatic}, our nonlinear dynamic control method introduces additional degrees of freedom through extra controller gains, which effectively mitigate potential instabilities arising from system nonlinearities. This flexibility expands the set of feasible control signals, resulting in tighter closed-loop intervals. Additionally, we establish a separation principle, allowing for design of observer and control gains, which facilitates the derivation of tractable programs for control gain design.
  
 \section{Preliminaries}
 
 {\emph{{Notation}.}} $\mathbb{R}^n,\mathbb{R}^{n  \times p},\mathbb{D}_n,\mathbb{N},\mathbb{N}_n$ denote the $n$-dimensional Euclidean space and the sets of $n$ by $p$ matrices, $n$ by $n$ diagonal matrices, natural numbers and natural numbers up to $n$, respectively. 
 %while $\mathbb{M}_n$ %and $\tilde{\mathbb{M}}_n$ 
 %denote\yong{s} the set of all $n$ by $n$ Metzler\footnote{A Metzler matrix is a square matrix in which all the off-diagonal components are nonnegative (equal to or greater than zero).} \mo{matrices}. %and M-matrices\footnote{An M-matrix is a square matrix in which all the off-diagonal components are nonpositive (equal to or lesser than zero).}, respectively. 
 For $M \in \mathbb{R}^{n \times p}$, $M_{i,j}$ denotes $M$'s entry in the $i$'th row and the $j$'th column, $M^\oplus\triangleq \max(M,\mathbf{0}_{n,p})$, $M^\ominus=M^\oplus-M$ and $|M|\triangleq M^\oplus+M^\ominus$, where $\mathbf{0}_{n,p}$ is the zero matrix in $\mathbb{R}^{n \times p}$ and $\mathbf{0}_{n}$ is the zero vector in $\mathbb{R}^n$, \yong{while $\textstyle{\mathrm{sgn}}(M) \in \mathbb{R}^{n \times p}$ is the element-wise sign of $M$ with $\textstyle{\mathrm{sgn}}(M_{i,j})=1$ if $M_{i,j} \geq 0$ and $\textstyle{\mathrm{sgn}}(M_{i,j})=-1$, otherwise.} %Furthermore, $\textstyle{\mathrm{sgn}}(M) \in \mathbb{R}^{n \times p}$, where $\textstyle{\mathrm{sgn}}(M)_{i,j}=1$ if $M_{i,j} \geq 0$ and $\textstyle{\mathrm{sgn}}(M)_{i,j}=0$, otherwise. 
 Furthermore, if $p=n$, 
 %$M^\text{d}$ denotes a diagonal matrix whose diagonal coincides with the diagonal of $M$, $M^\text{nd} \triangleq M-M^\text{d}$ and $M^{\text{m}} \triangleq M^\text{d}+|M^\text{nd}|$,
 \mk{$M \succ \mathbf{0}_{n,n}$ and $M \prec \mathbf{0}_{n,n}$ (or $M \succeq \mathbf{0}_{n,n}$ and $M \preceq \mathbf{0}_{n,n}$) denote that $M$ is positive and negative (semi-)definite, respectively}. Furthermore, $\tdiag(A_1,\dots, A_n)$ denotes a block diagonal matrix with $A_1, \dots,A_n$ being its diagonal block matrix entries, while $\tdiag_n(A)\triangleq \tdiag(\underbrace{A,\dots,A}_{n \ \text{times}})$. Finally $I_n$ denotes the identity matrix in $\mathbb{R}^n$. %is a block diagonal matrix whose diagonal is  
%\emph{Notation.} %We first summarize some notations used throughout the paper. 

Next, we introduce some useful definitions and results.
\begin{defn}[Interval]\label{defn:interval}
{An (multi-dimensional) interval {$\mathcal{I} \triangleq [\underline{s},\overline{s}]  \subset %%\in 
\mathbb{R}^n$} is the set of all real vectors $x \in \mathbb{R}^n$ that satisfies $\underline{s} \le x \le \overline{s}$, where $\underline{s}$, $\overline{s}$ and $\|\overline{s}-\underline{s}\|\mk{_{\infty}\triangleq \max_{i \in \{1,\cdots,n\}}s_i}$ are called minimal vector, maximal vector and \mk{interval} width of $\mathcal{I}$, respectively}. \yong{An interval matrix can be defined similarly.} %, in an element-wise manner.}
\end{defn}
\begin{defn}[Jacobian Sign-Stability] \label{defn:JSS}
A mapping $f :\mathcal{X} \subset \mathbb{R}^{n} \to  \mathbb{R}^{p}$ is (generalized) Jacobian sign-stable (JSS), if its (generalized) Jacobian matrix entries \fa{do} not change signs on its domain, i.e., if either of the following hold: 
\begin{align*}
&\forall x \in \mathcal{X}, \forall i \in \mathbb{N}_p,\forall j \in \mathbb{N}_n , J_f(x)_{i,j} \geq 0 \ \text{(positive JSS)}\\  
&\forall x \in \mathcal{X}, \forall i \in \mathbb{N}_p,\forall j \in \mathbb{N}_n , J_f(x)_{i,j} \leq 0 \  \text{(negative JSS)},
\end{align*} 
where $J_f(x)$ denotes the Jacobian matrix of $f$ at $x \in \mathcal{X}$. 
\end{defn}
\begin{prop}[Mixed-Monotone Decomposition]
\cite[Proposition 2]{9867741}\label{prop:JSS_decomp}
Let $f :\mathcal{X} \subset \mathbb{R}^{n} \to  \mathbb{R}^{p}$ and suppose $\forall x \in \mathcal{X}, J_f(x) \in [\underline{J}_f,\overline{J}_f]$, where $\underline{J}_f,\overline{J}_f$ are known matrices in $\mathbb{R}^{p \times n}$. Then, $f$ can be decomposed into a (remainder) affine mapping $Hx$ and a JSS mapping $\mu (\cdot)$, in an additive form: 
\begin{align}\label{eq:JSS_decomp}
\forall x \in \mathcal{X},f(x)=\mu(x)+Hx,
\end{align}
 where $H$ is a matrix in $\mathbb{R}^{p \times n}$, that satisfies the following 
 \begin{align}\label{eq:H_decomp}
 \forall (i,j) \in \mathbb{N}_p \times \mathbb{N}_n, H_{i,j}=(\overline{J}_f)_{i,j} \ \lor \ H_{i,j}=(\underline{J}_f)_{i,j}.    
 \end{align}
\end{prop}
\begin{prop}[Tight Decomposition Functions for JSS Mappings {\cite[Proposition 4 \& Lemma 3]{9867741}}]\label{prop:tight_decomp}
Let $\mu:{\mathcal{Z}} \subset \mathbb{R}^{n_z} \to \mathbb{R}^p$ be a JSS mapping on its domain. Then, it admits a tight decomposition function for each $\mu_i,\ i \in \mathbb{N}_p$ as follows: %that has 
% with the following form: %ordered pair $x_2 \leq x_1$: %$\lambda_d$ (cf. Proposition \ref{prop:mm_dec}) for any ordered pair $\underline{x}, \overline{x} \in \mathcal{X}, \underline{x} \leq \overline{x}$, can be tightly computed as follows:
\begin{align}\label{eq:JJ_decomp}
%\forall i \in \mathbb{N}_p, 
\mu_{d,i}({z}_1,{z}_2)\hspace{-.1cm}=\hspace{-.1cm}\mu_i(D^i{z}_1\hspace{-.1cm}+\hspace{-.1cm}(I_{n_z}\hspace{-.1cm}-\hspace{-.1cm}D^i){z}_2), 
\end{align}
{for any ordered ${z_1, z_2 \in \mathcal{Z}}$}, where $D^i$ %\in \mathbb{D}_\md{n_z}$ 
is a binary diagonal matrix determined by which vertex of the interval {$[{z}_2,{z}_1]$ {or $[z_1,z_2]$}} that maximizes {(if $z_2 \leq z_1$) or minimizes {(if $z_2 > z_1$)}} the function 
$\mu_i$ that can be found in closed-form as: %and \md{is computed} as follows:
\begin{align}\label{eq:Dj}
D^i=\textstyle{\mathrm{diag}}(\max(\textstyle{\mathrm{sgn}}(\yong{\overline{J}^{\mu}_i}),\mathbf{0}_{1,{n_z}})).
\end{align}
Furthermore, for any interval domain $[\underline{z} \ \overline{z}]\subseteq \mathcal{Z}$, the following inequality holds:
\begin{align}\label{eq:mm_bounding}
\mu_d(\overline{z},\underline{z})-\mu_d(\underline{z},\overline{z})\leq F_{\mu}(\overline{z}-\underline{z}), 
\end{align}
where $\mu_d=[\mu_{d,1},\dots, \mu_{d,p}]^\top$ and 
\begin{align}\label{eq:F_bounding}
F_{\mu}=\overline{J}^\oplus_{\mu}+\underline{J}^\ominus_{\mu}.
\end{align}
Finally, if $\mu$ is computed through the mixed-monotone decomposition in Proposition \ref{prop:JSS_decomp}, then
\begin{align}\label{eq:JSS_Jacobian}
\overline{J}^\mu=\overline{J}^f-H \ \text{and} \ \underline{J}^\mu=\underline{J}^f-H.
\end{align}  
\end{prop}
\section{Problem Formulation} \label{sec:Problem}
%\vspace{-0.1cm}
\noindent\textbf{\emph{System Assumptions.}} 
Consider the following nonlinear discrete-time (DT) system:  
\begin{align} \label{eq:system}
\begin{array}{ll}
\mathcal{G}: \begin{cases} {x}_{k+1} &= f(x_k)+Bu_k+Ww_k    \\
                                              y_k &= g(x_k)+Du_k+Vv_k
                                              %\\
%                                              \  y_c\in g_c(x,u_c,\mathcal{V}), & x \in C, \\ 
%                                              \ y_d \in g_d(x,u_d,\mathcal{V}), & x \in D, 
                                              \end{cases},
\end{array}\hspace{-0.2cm}
\end{align}
%where $x_t^+=\dot{x}_t, \yong{\mathbb{T}} = \mathbb{R}_{\ge 0}$ if $\mathcal{G}$ is a CT and $x_t^+=x_{t+1}, \yong{\mathbb{T}}= \{0\}\cup \mathbb{N}$, if $\mathcal{G}$ is a DT system. 
where $x_k \in \mathcal{X} \subset \mathbb{R}^n$, $u_k \in \mathbb{R}^m$, $y_k \in \mathbb{R}^l$, $w_k \in \mathcal{W} \triangleq [\underline{w},\overline{w}] \subset \mathbb{R}^{n_w}$ and $v_k \in \mathcal{V} \triangleq [\underline{v},\overline{v}] \subset \mathbb{R}^{n_v}$ are state, control input, output (measurement), process disturbance and measurement noise signals, respectively. Furthermore, ${f}:\mathcal{X}  \to \mathbb{R}^n$ and ${g}:\mathcal{X}  \to \mathbb{R}^l$  are known, nonlinear vector fields, and $B \in \mathbb{R}^{n} \times \mathbb{R}^{n}, W \in \mathbb{R}^{n} \times\mathbb{R}^{n_w},D \in \mathbb{R}^{l} \times\mathbb{R}^{b}$ and $V \in \mathbb{R}^{l} \times\mathbb{R}^{n_v}$ are known matrices. We assume the following. 
\begin{assumption}\label{ass:mixed_monotonicity}
 The initial state $x_0$ satisfies $x_0 \in \mathcal{X}_0 = [ \underline{x}_0,\overline{x}_0]$, where $\underline{x}_0$ and $\overline{x}_0$ {are} known initial state bounds, and the values of the output/measurement $y_k$ signals are known at all times. Moreover, the mappings $f$ and $g$ are known, differentiable and have bounded Jacobians, i.e., satisfy
 \begin{align*}%\label{eq:jac_bounds}
 \underline{J}_{f} \leq J_{\nu}(x) \leq \overline{J}_{\nu}, \ \forall \nu \in \{f,g\}, \ \forall x \in \mathcal{X},
 \end{align*}
 where $J_f$ and $J_g$ are the Jacobian matrix functions and $\overline{J}_{f},\underline{J}_{f} \in \mathbb{R}^{n \times n}, \overline{J}_{g},\underline{J}_{g} \in \mathbb{R}^{l \times n}$ are  priori known matrices.
\end{assumption}
Based on Assumption \ref{ass:mixed_monotonicity} and by leveraging Proposition \ref{prop:JSS_decomp}, the matrices $A \in \mathbb{R}^{n \times n}$ and $C \in \mathbb{R}^{l \times n}$ are chosen such that the following decompositions hold (cf. Definition \ref{defn:JSS}):
\begin{align} \label{eq:JSS_decom}
%\begin{array}{lr}
\hspace{-.2cm}\forall x \in \mathcal{X}: \begin{cases}f(x)=Ax+\phi(x) \\ g(x)=Cx+\psi(x) \end{cases} \hspace{-.4cm} s.t. \ \phi,\psi \ \text{are JSS in} \ \mathcal{X}.
%\end{array}
\end{align} 
Recall that according to Proposition \ref{prop:tight_decomp}, the JSS mappings $\phi$ and $\psi$ admit matrices $F_{\phi}$ and $F_{\psi}$ that can be computed via \eqref{eq:F_bounding} and satisfy \eqref{eq:mm_bounding}\footnote{The considered class of nonlinear systems in \eqref{eq:system} can be easily extended to be nonlinear in the input $u_k$ and noise signals $w_k,v_k$, i.e., $f(x_k,u_k,w_k)$, and $g(x_k,u_k,v_k)$, by applying the decomposition in \eqref{eq:JSS_decom} to $f$ and $g$ and decompose them into linear functions in and JSS mappings.}. Without loss of generality, we assume $F_{\phi}$ is invertible. The reason is that by construction, $F_{\phi}$ is a non-negative matrix (cf. \eqref{eq:F_bounding}). So, even if the initially computed $F_{\phi}$ is not invertible, one can increase its diagonal elements until the modified $F_{\phi}$ becomes a diagonally dominant, and so invertible matrix~\cite{sootla2017block}. Furthermore, since the modified $F_{\phi}$ is greater than or equal to the initial one, it still satisfies \eqref{eq:mm_bounding} given the non-negativity of all items in both sides of the inequality.
%We also assume the following.
%Based on Assumption \ref{} and via applyin Proposition \ref{}, the a mapping $f$ can be decomposed as $f(x)=Ax+\phi(x), \forall x \in \mathcal{X}$, where $\phi$ is a JSS mapping.
%\begin{assumption} \label{ass:initial_interval}
% The initial state $x_0$ satisfies $x_0 \in \mathcal{X}_0 = [ \underline{x}_0,\overline{x}_0]$, where $\underline{x}_0$ and $\overline{x}_0$ {are} known initial state bounds, and the values of the output/measurement $y_k$ signals are known at all times. 
% \end{assumption} 
 Our goal is to design the control input signal $u_k$ as a (implicit) function of system measurements/outputs to stabilize the closed-loop DT nonlinear system \eqref{eq:system}--\eqref{eq:JSS_decom} in an optimal manner. Due to the existing uncertainties, we do not assume that the state is measurable (even the initial state is uncertain), and we only measure the output $y_k$.
\begin{problem}\label{eq:stab_synth}
Given the nonlinear system in \eqref{eq:system}--\eqref{eq:JSS_decom}, as well as Assumption \ref{ass:mixed_monotonicity}, 
synthesize a dynamic control system for the plant $\mathcal{G}$ so that the resulting closed-loop system is stable, and the design is optimal in the sense of $\mathcal{H}_{\infty}$ controller synthesis, guaranteeing the uniform boundedness of the worst-case states and minimizing the effect of the disturbance signals. 
\end{problem}
\section{Proposed Control Design}
Before discussing our control design strategy, we first briefly introduce the notions of framer and interval observer which we extensively use throughout the paper.
\begin{defn}[Correct Interval \mk{Framers}]\label{defn:framers}
%Suppose Assumptions \ref{ass:initial_interval}--\ref{ass:known_input_output} hold. 
Given the nonlinear plant $\mathcal{G}$, %with a time domain $E \triangleq \bigcup_{j=0}^{J-1}([t_j,t_{j+1}],j)$, 
the sequences of signals $\{\overline{x}_k,\underline{x}_{k}\}_{k=0}^{\infty}$ in $\mathbb{R}^n$ are called upper and lower framers for the states of \eqref{eq:system}, if 
\begin{align}\label{eq:correctness}
\forall k \in [0,\infty), \ \underline{x}_k \leq x_k \leq \overline{x}_k.
\end{align}
% where $T \triangleq \{0\} \cup \mathbb{N}$ and $x(t) \triangleq x_t$ if $\mathcal{G}$ is a discrete-time system and $E\triangleq \mathbb{R}_+$ if $\mathcal{G}$ is a continuous-time system. 
In other words, starting from the initial interval $\underline{x}_0 \leq x_0 \leq \overline{x}_0$, the true state of the system in \eqref{eq:system}, $x_k$, is guaranteed to evolve within the interval flow-pipe $[\underline{x}_k,\overline{x}_k]$, for all $k \geq 0$. Finally, any dynamical system $\hat{\mathcal{G}}$ whose states are correct framers for the states of the plant $\mathcal{G}$, i.e., any (tractable) algorithm that returns upper and lower framers for the states of plant $\mathcal{G}$ is called a \emph{correct} interval \mk{framer} system for \eqref{eq:system}. 
\end{defn}
\begin{defn}[\mk{Framer} Error]\label{defn:error}
Given %a CHIO returning 
state framers \mk{$\underline{x}_k \leq \overline{x}_k, k \geq 0$}, $\varepsilon_k \triangleq \overline{x}_k-\underline{x}_k$, \mk{whose infinite norm denotes} the interval width \mk{of $[\underline{x}_k,\overline{x}_k]$ (cf. Definition \ref{defn:interval})},  
% defined as $\varepsilon_t \triangleq \overline{x}_t-\underline{x}_t$, 
 is called the \mk{framer} error. It can be easily verified that %with defining $\varepsilon_0 \triangleq \overline{x}_0-\underline{x}_0$, 
correctness (cf. Definition \ref{defn:framers}) implies that $\varepsilon_k \geq 0, \forall k \geq 0.$  
\end{defn}
\begin{defn}[Stability and \mk{Interval Observer}]\label{defn:stability}
A corresponding interval \mk{framer} to system \eqref{eq:system} is input-to-state stable (ISS), if the \mk{framer} error sequence $\{\varepsilon_{k}\}_{k=0}^{\infty}$ is bounded as follows: 
\begin{align}
\|\varepsilon_{k}\|_2 \leq \beta(\|\varepsilon_0\|,k)+\rho(\|\delta\|_{\ell_{\infty}}), \forall k \geq 0,
\end{align}
where $\delta \triangleq [\delta^\top_w \ \delta^\top_v]^\top$, $\beta$ and $\rho$ are functions of classes $\mathcal{KL}$\footnote{A continuous function $\alpha: [0,a) \to \mathbb{R}_{\ge 0}$ belongs to class $\mathcal{K}$ if it is strictly increasing with $\alpha(0)=0$. Moreover, the function $\alpha$ belongs to class $\mathcal{K}_{\infty}$ if it belongs to class $\mathcal{K}$, $a=\infty$, and $\lim_{r \to \infty}\alpha(r)=\infty$, i.e., $\alpha$ is unbounded. Finally,
A continuous function $\lambda : {[0,a) \times [0,\infty)} \to \mathbb{R}_{\ge 0}$ is {considered to be in} class $\mathcal{KL}$ if, for every fixed $t \geq 0$, {the function} $\lambda(s, t)$ {is in} class $\mathcal{K}$; for each fixed $s \geq 0$, $\lambda(s, t)$ {decreases with respect to $t$ and satisfies $\lim_{t \to \infty} \lambda(s, t) = 0$}} and $\mathcal{K}_{\infty}$, respectively, and
$\|\delta\|_{\ell_{\infty}}=\sup_{k \geq 0} \|\delta_k\|_2=\|\delta\|_2$ is the $\ell_{\infty}$ signal norm.
An ISS interval framer is called an interval observer.
\begin{defn}[$\mathcal{H}_{\infty}$-Robust \& Optimal Interval Observer]\label{defn:L_1} 
An interval framer $\hat{\mathcal{G}}$ is  $\mathcal{H}_{\infty}$-robust and optimal, if the  $\mathcal{H}_{\infty}$-gain of the framer error system $\tilde{\mathcal{G}}$, defined below, is minimized:
%minimizes $\gamma$ between all the observers belonging to that class, where 
\begin{align}\label{eq:L1_Def}
\|\tilde{\mathcal{G}}\|_{\ell_2} \triangleq \sup_{\|\delta\|_{\ell_2}=1} \|e\|_{\ell_2},
% (L^*,\gamma^*) \in \argmin\limits_{\{L,\gamma\}} \gamma \quad \text{s.t.} \ \frac{\lim\limits_{t \to \infty}\|\varepsilon_t\|_{1}}{\|\Delta\|_{1}} \leq \gamma.   
\end{align}
where $\|\nu\|_{\ell_2} \triangleq  \sqrt{\sum_{k=0}^{\infty}(\|\nu_k\|^2_2)}$ is the $\ell_2$ signal norms for $\nu \in \{e,\delta\}$. 
\end{defn}
%\begin{align}
%\forall (t,j) \in E, \|e(t,j)\| \leq \beta(\|\varepsilon_0\|,t+j)+\rho(\Delta),
%$\lim_{t \to \infty} \|\varepsilon _t\|=0$. \mk{A stable interval framer is called an interval observer.}
%\end{align}
%where $\|\varepsilon _t\|_{\tilde{p}}$ is $ \{1,2,\infty,\dots\}\ni \tilde{p}$-norm of the signal/mapping $\varepsilon _t$.
% for all $(t, j) \in \dom e$, 
%where $\Delta \triangleq \max\{\|\Delta w \|_{\infty},\|\Delta y_c \|_{\infty},\|\Delta y_d \|_{\infty}\}$, $\beta \in \mathcal{KL}$ and $\rho \in \mathcal{K}_{\infty}$.
\end{defn}
Our proposed control strategy involves the following steps:
%\begin{enumerate}[(i)]

\noindent \textbf{(i) Constructing a Robust Framer System.} Given that the actual system state in \eqref{eq:system}--\eqref{eq:JSS_decom} is uncertain and the state is not directly measured, we first design a robust framer system that, using the observed output signal, provides upper and lower bounds on the actual system state. This framer system, constructed to handle any control input $u_k$, specified observer gain, and any bounded noise and disturbance realizations, yields reliable interval estimates of the system state.\\%\label{item:framer}   
   % Since the actual state of the system in \eqref{eq:system}--\eqref{eq:JSS_decom} is uncertain and not measured, we first construct a corresponding certain and robust \emph{framer system}, that given the output (measurement) signal and by construction, returns upper and lower interval-valued estimates of the state of the actual system, given \emph{any} control signal $u_k$, any to-be-designed observer gain, and for any realization of the bounded noise and disturbance signals.\label{item:framer}

\vspace{-.2cm}
\noindent \textbf{(ii) Designing a Dynamic Control Strategy for the Framer System.} Next, we develop a dynamic control strategy for the framer system created in step (i) to ensure stability of the actual system through stabilizing the framers.\\%\label{item:control}  
    %Then, we consider a to-be-designed dynamic control strategy for the framer system constructed in \eqref{item:framer}. The goal is to stabilize the framer system which results in a stable actual system.\label{item:control}

\vspace{-.2cm}
\noindent \textbf{(iii) Establishing the Closed-Loop Framer System.} By applying the control inputs from step (ii) to the framer system from step (i) and augmenting the framer and the dynamic controller systems, we derive the closed-loop framer system. We stabilize this closed-loop system by designing the observer and dynamic control independently, facilitated by a separation principle. This results in an ISS closed-loop configuration for the actual plant.  
%\vspace{-.6cm}
\subsection{Interval Framers} \label{sec:obsv}
Given the nonlinear plant $\mathcal{G}$, in order to address Problem \ref{eq:stab_synth}, we first propose an interval observer (cf. Definition \ref{defn:framers}) for $\mathcal{G}$ similar to \cite{9790824}, through the following dynamical system:

\vspace{-.3cm}
{\small
\begin{align}\label{eq:observer}
%\hat{\mathcal{G}}:\begin{cases}
\begin{array}{rl}
\overline{x}_{k+1}\hspace{-.2cm}&=\hspace{-.1cm}(A\hspace{-.1cm}-\hspace{-.1cm}LC)^\oplus \overline{x}_k\hspace{-.1cm}-\hspace{-.1cm}(A\hspace{-.1cm}-\hspace{-.1cm}LC)^\ominus \underline{x}_k\hspace{-.1cm}+\hspace{-.1cm}(LV)^\ominus \overline{v}\hspace{-.1cm}-\hspace{-.1cm}(LV)^\oplus \underline{v} \\
\hspace{-.1cm}&+\phi_d(\overline{x}_k,\underline{x}_k)\hspace{-.1cm}+\hspace{-.1cm}Ly_k\hspace{-.1cm}+\hspace{-.1cm}(B\hspace{-.1cm}-\hspace{-.1cm}LD)u_k\hspace{-.1cm}+\hspace{-.1cm}W^{\oplus}\overline{w}\hspace{-.1cm}-\hspace{-.1cm}W^{\ominus}\underline{w}\\
\hspace{-.1cm}&+L^{\ominus}\psi_d(\overline{x}_k,\underline{x}_k)\hspace{-.1cm}-\hspace{-.1cm}L^{\oplus}\psi_d(\underline{x}_k,\overline{x}_k),\\ \vspace{.1cm}
\underline{x}_{k+1}\hspace{-.2cm}&=\hspace{-.1cm}(A\hspace{-.1cm}-\hspace{-.1cm}LC)^\oplus \underline{x}_k\hspace{-.1cm}-\hspace{-.1cm}(A\hspace{-.1cm}-\hspace{-.1cm}LC)^\ominus \overline{x}_k\hspace{-.1cm}+\hspace{-.1cm}(LV)^\ominus \underline{v}\hspace{-.1cm}-\hspace{-.1cm}(LV)^\oplus \overline{v} \\
\hspace{-.1cm}&+L^{\ominus}\psi_d(\underline{x}_k,\overline{x}_k)\hspace{-.1cm}-\hspace{-.1cm}L^{\oplus}\psi_d(\overline{x}_k,\underline{x}_k)\\
\hspace{-.1cm}&+\phi_d(\underline{x}_k,\overline{x}_k)\hspace{-.1cm}+\hspace{-.1cm}Ly_k\hspace{-.1cm}+\hspace{-.1cm}(B\hspace{-.1cm}-\hspace{-.1cm}LD)u_k\hspace{-.1cm}+\hspace{-.1cm}W^{\oplus}\underline{w}\hspace{-.1cm}-\hspace{-.1cm}W^{\ominus}\overline{w},
 \end{array}
\end{align}
}

\noindent where $\phi_d:\mathbb{R}^n \times \mathbb{R}^n \to \mathbb{R}^n$ and $\psi_d:\mathbb{R}^n \times \mathbb{R}^n \to \mathbb{R}^l$ are tight mixed-monotone decomposition functions of the JSS mappings $\phi$ and $\psi$, computed via Proposition~\ref{prop:tight_decomp}. Moreover, $L \in \mathbb{R}^{n \times l}$ is a to-be-designed observer gain matrix. 
%designed via Theorem \ref{thm:stability}, such that the proposed observer $\hat{\mathcal{G}}$ possesses the desired properties discussed in the following subsections. 
Defining $\varepsilon_k \triangleq \overline{x}_k-\underline{x}_k$, we obtain the following system from \eqref{eq:observer} that governs 
 the dynamics of the framer errors:
\begin{align}\label{eq:farmer_error}
 \varepsilon_{k+1}\hspace{-.1cm}=|A-LC|\varepsilon_k+\delta^{\phi}_k+\hspace{-.1cm}|L|\delta^{\psi}_k+|LV|\delta_v+|W|\delta_w,   
\end{align}
where $\delta_{\alpha} \triangleq \overline{\alpha}_k-\underline{\alpha}_k, \forall \alpha \in \{w,v\}$, while
\begin{align}\label{eq:deltas}
%\begin{array}{rl}
{\delta}^{\rho}_k \triangleq \rho_d(\overline{x}_k,\underline{x}_k)-\rho_d(\underline{x}_k,\overline{x}_k), \ \forall \rho \in \{\phi,\psi\}.
%\end{array}
\end{align}
Proposition \eqref{prop:frame}, %followed by a short proof due to the space limit, 
ensures the framer property of \eqref{eq:observer} %is an interval framer 
for \eqref{eq:system}.%--\eqref{eq:JSS_decom}.
\begin{prop}\label{prop:frame}%\cite[Theorems 1 \& 2]{9790824}
The system in \eqref{eq:observer} constructs a framer system for the plant in \eqref{eq:system}--\eqref{eq:JSS_decom} for all values of the control signal $u_k$, all realizations of the bounded noise and disturbance $v_k,w_k$ and all values of the gain $L$. Consequently, \eqref{eq:farmer_error} is a positive system. Furthermore, if $L$ is a solution to the SDP in \cite[Equations (17) and (19)]{9790824}, then \eqref{eq:farmer_error} is an ISS system, i.e., \eqref{eq:observer} is an interval observer for \eqref{eq:system}--\eqref{eq:JSS_decom}.
\end{prop}
\begin{proof}  
The proof follows the lines of~\cite[Theorems 1 \& 2]{9790824} with the slight modification that the terms $Bu_k$ and $Du_k$ are added to the state and output equations, respectively, and to the corresponding framers, since they are treated as known variables when upper and lower framers are computed. 
\end{proof}
\subsection{Closed-Loop System and Separation Principle}\vspace{-0.05cm}
In this subsection, we use the robust state estimates, i.e., the interval framers in \eqref{eq:observer}, for $\mathcal{H}_{\infty}$ controller synthesis, in the sense of guaranteeing the uniform boundedness of the worst-case states and minimizing the effect of the disturbance signals. To this end, we consider a nonlinear dynamic controller with disturbance rejection terms: %in the following
%form: 
\begin{align}\label{eq:control}
\begin{array}{rl}
 x^c_{k+1}&=A_cx^c_k+B_cy^c_k+K^x_{\nu}(\phi_d(\overline{x},\underline{x})-\phi_d(\underline{x},\overline{x}))\\
 u_k&=C_cx^c_k+D_cy^c_k+K^u_{\nu}(\phi_d(\overline{x},\underline{x})-\phi_d(\underline{x},\overline{x})),
 \end{array}
\end{align}
%for the known framer system in \eqref{eq:observer}, 
where $y^c_k=[\underline{x}^\top_k \ \overline{x}^\top_k]^\top$, and $A_c \in \mathbb{R}^{n \times n},B_c=[-\underline{K}_b \ \overline{K}_b] \in \mathbb{R}^{n \times 2n}, C_c \in \mathbb{R}^{m \times n}, D_c=[-\underline{K}_d \ \overline{K}_d] \in \mathbb{R}^{m \times 2n}, K^x_{\nu} \in \mathbb{R}^{n \times n},K^u_{\nu} \in \mathbb{R}^{m \times n}$ are to-be-designed dynamic controller gain matrices. The goal is to stabilize the framer system \eqref{eq:observer} which consequently stabilizes the trajectory of the actual closed-loop system \eqref{eq:system}--\eqref{eq:JSS_decom}, as well as to minimize the effect of the augmented noise in the sense of minimizing the $\mathcal{H}_{\infty}$ norm of the closed-loop error system.%, given the fact the states of the former bound/frame the states of the latter from below and above by Lemma \ref{lem:frame}. 
\begin{rem} 
The nonlinear feedback components of the proposed strategy compensate for the potential dynamic instabilities that are due to the nonlinearity of the system and extend the feasible set of admissible control signals, hence potentially resulting in tighter closed-loop intervals compared to linear control approaches, for example, in \cite{efimov2013control}. Moreover, it is worth emphasizing that we are proposing a dynamic control strategy that by design can be considered as a generalization of and improvement to the static feedback control we proposed in our previous work \cite{khajenejad2024optimalcontrolstatic}, (see the comparisons in the Section \ref{sec:example}). 
\end{rem}
\vspace{-.2cm}
In order to obtain the dynamics of the closed-loop system, we plug $u_k$ from \eqref{eq:control} into \eqref{eq:observer} and use the fact that 
\begin{align}%\label{eq:auxx}
\begin{array}{c}
Ly_k+(B-LD)u_k=Bu_k+L(y_k-Du_k)\\
=Bu_k+L(Cx_k+\psi(x_k)+Vv_k),
\end{array}
\end{align}
which results in the following closed-loop system: %for %\eqref{eq:observer}: 

\vspace{-.3cm}
{\small
\begin{align} \label{eq:eqiv_sys}
%\begin{align}\label{eq:observer}
%\hat{\mathcal{G}}:\begin{cases}
\begin{array}{rl}
\overline{x}_{k+1}\hspace{-.1cm}&=\hspace{-.1cm}((A\hspace{-.1cm}-\hspace{-.1cm}LC)^\oplus\hspace{-.1cm}+\hspace{-.1cm}B\overline{K}_d) \overline{x}_k\hspace{-.1cm}-((A\hspace{-.1cm}-\hspace{-.1cm}LC)^\ominus\hspace{-.1cm}+\hspace{-.1cm}B\underline{K}_d) \underline{x}_k\\
&+(LV)^\ominus \overline{v}\hspace{-.1cm}-\hspace{-.1cm}(LV)^\oplus \underline{v}\hspace{-.1cm}+\hspace{-.1cm}LVv_k\hspace{-.1cm}+\hspace{-.1cm}W^{\oplus}\overline{w}\hspace{-.1cm}-\hspace{-.1cm}W^{\ominus}\underline{w}  \\
&+\phi_d(\overline{x}_k,\underline{x}_k)\hspace{-.1cm}+\hspace{-.1cm}L^{\oplus}\overline{\delta}^{\psi}_k\hspace{-.1cm}+\hspace{-.1cm}L^{\ominus}\underline{\delta}^{\psi}_k\hspace{-.1cm}+\hspace{-.1cm}BK^u_{\nu}\delta^{\phi}_k\hspace{-.1cm}+\hspace{-.1cm}LCx_k,\\
\underline{x}_{k+1}\hspace{-.1cm}&=((A\hspace{-.1cm}-\hspace{-.1cm}LC)^\oplus\hspace{-.1cm}-\hspace{-.1cm}B\underline{K}_d) \underline{x}_k\hspace{-.1cm}-\hspace{-.1cm}((A\hspace{-.1cm}-\hspace{-.1cm}LC)^\ominus\hspace{-.1cm}-\hspace{-.1cm}B\overline{K}_d) \overline{x}_k\\
&+(LV)^\ominus \underline{v}\hspace{-.1cm}-\hspace{-.1cm}(LV)^\oplus \overline{v}\hspace{-.1cm}+\hspace{-.1cm}LVv_k\hspace{-.1cm}+\hspace{-.1cm}W^{\oplus}\underline{w}\hspace{-.1cm}-\hspace{-.1cm}W^{\ominus}\overline{w}  \\
&+\phi_d(\underline{x}_k,\overline{x}_k)\hspace{-.1cm}-\hspace{-.1cm}L^{\oplus}\overline{\delta}^{\psi}_k\hspace{-.1cm}-\hspace{-.1cm}L^{\ominus}\underline{\delta}^{\psi}_k\hspace{-.1cm}+\hspace{-.1cm}BK^u_{\nu}\delta^{\phi}_k\hspace{-.1cm}+\hspace{-.1cm}LCx_k,
 \end{array}
\end{align}
}

\vspace{-.2cm}
\noindent where ${\delta}^{\phi}_k,{\delta}^{\psi}_k$ are defined in \eqref{eq:deltas}, $\overline{\delta}^{\rho}_k \triangleq \rho_d(\overline{x}_k,\underline{x}_k)-\rho(x_k)$, and $\underline{\delta}^{\rho}_k \triangleq \rho(x_k)-\rho_d(\underline{x}_k,\overline{x}_k), \forall \rho \in \{\phi,\psi\}$. The system in \eqref{eq:eqiv_sys} still contains the actual state $x_k$ which is not measured and hence is not accessible. To resolve this issue, we apply a change if variables and define the following sequences of \emph{upper and lower closed-loop errors}: %for $k \geq 0$: %via the following change of variables:
\begin{align}\label{eq:frame_error}
\begin{array}{rl}
\overline{e}_k &\triangleq \overline{x}_k-x_k \Rightarrow x_k=\overline{x}_k-\overline{e}_k, \\
\underline{e}_k &\triangleq x_k-\underline{x}_k \Rightarrow x_k=\underline{x}_k+\underline{e}_k.
\end{array}
\end{align}
 %First, we define the upper and lower framer errors:
Combining \eqref{eq:frame_error} and \eqref{eq:system} and noting that: $$L(y_k-Cx_k-\psi(x_k)-Du_k-Vv_k)=0, \ \forall k \ge 0, \ \text{returns:}$$ %returns:
%\begin{align*}
$x_{k+1}=(A-LC)x_k-L\psi(x_k)+\phi(x_k)
+Ww_k+Ly_k-LVv_k+(B-LD)u_k$,
%\end{align*}
%The following lemma shows that the a
which, along with \eqref{eq:eqiv_sys} and the fact that $M=M^{\oplus}-M^{\ominus}$ for any matrix $M$ results in the dynamics of the closed-loop errors: %which if 
\begin{align}\label{eq:err_sys}
% \begin{align} \label{eq:eqiv_sys}
%\begin{align}\label{eq:observer}
%\hat{\mathcal{G}}:\begin{cases}
\begin{array}{rl}
\hspace{-.2cm}\overline{e}_{k+1}\hspace{-.1cm}&=(A\hspace{-.1cm}-\hspace{-.1cm}LC)^\oplus \overline{e}_k-(A\hspace{-.1cm}-\hspace{-.1cm}LC)^\ominus \underline{e}_k+\overline{\delta}^{\phi}_k-Ww_k\\
&+(LV)^\ominus \overline{v}\hspace{-.1cm}-\hspace{-.1cm}(LV)^\oplus \underline{v}+LVv_k+W^{\oplus}\overline{w}-\hspace{-.1cm}W^{\ominus}\underline{w},\\
%&+\phi_d(\overline{x}_k,\underline{x}_k)+L^{\oplus}\overline{\delta}^{\psi}_k+L^{\ominus}\underline{\delta}^{\psi}_k\\
\hspace{-.2cm}\underline{e}_{k+1}\hspace{-.1cm}&=(A\hspace{-.1cm}-\hspace{-.1cm}LC)^\oplus \underline{e}_k-(A\hspace{-.1cm}-\hspace{-.1cm}LC)^\ominus \overline{e}_k+\underline{\delta}^{\phi}_k+Ww_k\\
&+(LV)^\ominus \underline{v}\hspace{-.1cm}-\hspace{-.1cm}(LV)^\oplus \overline{v}+LVv_k+W^{\oplus}\underline{w}-\hspace{-.1cm}W^{\ominus}\overline{w},
 \end{array}
\end{align}   
%\end{align}
where $\overline{\delta}^{\phi}_k,\underline{\delta}^{\phi}_k$ are given under \eqref{eq:eqiv_sys}. We are ready to state our first result, which shows that the augmented system of closed-loop framers and errors has a comparison system that is independent of $x_k$, and can be stabilized by separately designing the observer gain $L$, and the dynamic control gain matrices $A_c,B_c,C_c,D_c,K^u_{\nu},K^x_{\nu}$. 
\begin{lem}[Separation Principle]\label{lem:separation}
The augmentation of \eqref{eq:eqiv_sys} and \eqref{eq:err_sys} can be stabilized by first designing a stabilizing observer gain $L$ for the framer error system \eqref{eq:farmer_error}, and then synthesizing the dynamic control gains $A_c,B_c,C_c,D_c,K^u_{\nu},K^x_{\nu}$, given $L$, to stabilize the augmented system.
\end{lem}
\begin{rem}
It is worth noting that Lemma \ref{lem:separation} may not be considered as a ``full" separation
principle in the sense that the interval observer gain $L$
remains an input for the synthesis of the dynamic control
gains, thus possibly influencing the resulting control
performance.
\end{rem}
\begin{proof}
Augmenting \eqref{eq:control}, \eqref{eq:eqiv_sys} and \eqref{eq:err_sys}, and applying the following inequalities for $\rho \in \{\phi,\psi\}$: 
\begin{align}\label{eq:F}
\begin{array}{rl}
\overline \delta^{\rho}_k&=\hspace{-.1cm}\rho_d(\overline{x}_k,\underline{x}_k)-\rho(x_k)\leq \rho_d(\overline{x}_k,\underline{x}_k)-\hspace{-.1cm}\rho_d(\underline{x}_k,\overline{x}_k)\\
&\leq F_{\rho}(\overline{x}_k-\underline{x}_k)=F_{\rho}(\overline{e}_k+\underline{e}_k),\\
%&=F_{\rho}(\overline{e}_k+\underline{e}_k),\\
\underline \delta^{\rho}_k&=\hspace{-.1cm}\rho(x_k)-\rho_d(\underline{x}_k,\overline{x}_k)\leq \rho_d(\overline{x}_k,\underline{x}_k)-\hspace{-.1cm}\rho_d(\underline{x}_k,\overline{x}_k)\\
&\leq F_{\rho}(\overline{x}_k-\underline{x}_k)=F_{\rho}(\overline{e}_k+\underline{e}_k),
%&=F_{\rho}(\overline{e}_k+\underline{e}_k),
\end{array}
\end{align}
according to Proposition \ref{prop:tight_decomp}, combined with the fact that $\rho_d(\underline{x}_k,\overline{x}_k)\leq \rho(x_k)\leq \rho_d(\overline{x}_k,\underline{x}_k)$ (by properties of decomposition functions),
returns a comparison augmented system:
\begin{align}\label{eq:compare_sys}
z_{k+1}\leq \tilde{A}z_k+\lambda(z_k)+\Lambda \eta_k, 
\end{align}
where $\eta_k \triangleq [\overline{w}^\top \underline{w}^\top w^\top_k \overline{v}^\top \underline{v}^\top v_k^\top]^\top$, $z\triangleq [\overline{e}^\top \underline{e}^\top x^{c\top}_k \overline{x}^\top \underline{x}^\top]^\top$,

\vspace{-.3cm}
{\small
\begin{align*}
 \lambda(z) &\triangleq [\mathbf{0}^\top_n \hspace{.1cm} \mathbf{0}^\top_n \hspace{.1cm} \mathbf{0}^\top_n \hspace{.1cm}  \phi^\top_d(\overline{x},\underline{x}) \hspace{.1cm} \phi^\top_d(\underline{x},\overline{x})]^\top,\\
 \Lambda &\triangleq \begin{bmatrix} W^\oplus & -W^\ominus &-W & (LV)^\ominus & -(LV)^\oplus & LV \\  -W^\ominus & W^\oplus & W & -(LV)^\oplus & (LV)^\ominus & LV \\
 \mathbf{0}_{n,n_w} & \mathbf{0}_{n,n_w} &\mathbf{0}_{n,n_w} & \mathbf{0}_{n,n} & \mathbf{0}_{n,n} & \mathbf{0}_{n,n} \\
 W^\oplus & -W^\ominus &\mathbf{0}_{n,n_w} & (LV)^\ominus & -(LV)^\oplus & LV \\  -W^\ominus & W^\oplus & \mathbf{0}_{n,n_w} & -(LV)^\oplus & (LV)^\ominus & LV
 \end{bmatrix},\\
\tilde{A}&=\begin{bmatrix} [\tilde{A}_{11}] & [\tilde{A}_{12}] & \mathbf{0}_{n,n} & \mathbf{0}_{n,n} & \mathbf{0}_{n,n} \\  [\tilde{A}_{21}] & [\tilde{A}_{22}] &\mathbf{0}_{n,n} & \mathbf{0}_{n,n} & \mathbf{0}_{n,n}\\
\mathbf{0}_{n,n} &  \mathbf{0}_{n,n} & [\tilde{A}_{33}] & [\tilde{A}_{34}] & [\tilde{A}_{35}] \\ 
[\tilde{A}_{41}] & [\tilde{A}_{42}] & [\tilde{A}_{43}] &[\tilde{A}_{44}] & [\tilde{A}_{45}] \\
[\tilde{A}_{51}] & [\tilde{A}_{52}] & [\tilde{A}_{53}] &[\tilde{A}_{54}] & [\tilde{A}_{55}] 
\end{bmatrix}.
\end{align*}
}

\noindent Moreover, the block matrices inside $\tilde{A}$ are as follows:

\vspace{-.3cm}
{\small
\begin{align*}
[\tilde{A}_{11}] &\triangleq (A-LC)^\oplus+F_{\phi}+|L|F_{\psi}, \ \hspace{-.1cm}[\tilde{A}_{33}]\triangleq A_c,\\
[\tilde{A}_{12}] &\triangleq (A-LC)^\ominus+F_{\phi}+|L|F_{\psi}, \ \hspace{-.1cm}[\tilde{A}_{34}]\triangleq \overline{K}_b+K^x_{\nu}F_{\phi},\\
[\tilde{A}_{21}] &\triangleq (A-LC)^\ominus+F_{\phi}-|L|F_{\psi}, \\
[\tilde{A}_{22}] &\triangleq (A-LC)^\oplus+F_{\phi}-|L|F_{\psi}, \ \hspace{-.1cm}[\tilde{A}_{43}] \triangleq BC_c,\\
[\tilde{A}_{41}] &\triangleq -LC+|L|F_{\psi}+BK^u_{\nu}F_{\phi},\ \ [\tilde{A}_{53}] \triangleq BC_c,\\ 
[\tilde{A}_{44}] &\triangleq (A-LC)^\oplus+\hspace{-.1cm}B\overline{K}_d+LC,\ \hspace{-.1cm}[\tilde{A}_{35}]\triangleq -(\underline{K}_b\hspace{-.1cm}+\hspace{-.1cm}K^x_{\nu}F_{\phi}),\\
[\tilde{A}_{45}] &\triangleq -(A-LC)^\ominus\hspace{-.1cm}-\hspace{-.1cm}B\underline{K}_d, \ [\tilde{A}_{51}] \triangleq -|L|F_{\psi}+BK^u_{\nu}F_{\phi},\\
 [\tilde{A}_{52}] &\triangleq LC-|L|F_{\psi}+BK^u_{\nu}F_{\phi},\ [\tilde{A}_{42}] \triangleq |L|F_{\psi}\hspace{-.1cm}+\hspace{-.1cm}BK^u_{\nu}F_{\phi},\\
 [\tilde{A}_{54}] &\triangleq -(A\hspace{-.1cm}-\hspace{-.1cm}LC)^\ominus\hspace{-.1cm}+\hspace{-.1cm}B\overline{K}_d, \
 [\tilde{A}_{55}]\hspace{-.1cm} \triangleq\hspace{-.1cm} (A\hspace{-.1cm}-\hspace{-.1cm}LC)^\oplus\hspace{-.1cm}-\hspace{-.1cm}\hspace{-.1cm}B\underline{K}_d\hspace{-.1cm}+\hspace{-.1cm}LC. 
\end{align*}
}

\noindent The comparison system \eqref{eq:compare_sys} has a linear component with state matrix $\tilde{A}$, as well as the Lipschitz nonlinear component $\lambda(z)$. Note that $\lambda$ is a locally Lipschitz mapping since $\phi$ has bounded Jacobin by Assumption \ref{ass:mixed_monotonicity} and Proposition \ref{prop:tight_decomp}, and hence is locally Lipschitz. Consequently, $\phi_d$ is Lipschitz by construction (cf. \eqref{eq:JJ_decomp} in Proposition \ref{prop:tight_decomp}). Since, $\tilde{A}$ is a block lower triangular matrix, its set of eigenvalues is a superset of the set of eigenvalues of the matrix $\tilde{A}_u \triangleq \begin{bmatrix} \tilde{A}_{11} & \tilde{A}_{12} \\ \tilde{A}_{21} & \tilde{A}_{22} \end{bmatrix}$ which only depend on $L$. Moreover, the nonlinear vector function $\lambda(z)$ does not affect the stability of $\tilde{A}_u$ due to its zero upper elements. Hence, the observer and controller gains can be designed separately to stabilize system \eqref{eq:compare_sys}. Moreover, by applying the last inequality in \eqref{eq:F}, it is straightforward to see that \eqref{eq:err_sys} admits a linear comparison system with the state matrix $\tilde{A}_u$. Hence, any observer gain that stabilizes \eqref{eq:compare_sys} should stabilize \eqref{eq:err_sys} and vice versa. Finally, from the definitions of $\varepsilon_k, \overline{e}_k, \underline{e}_k$, we have $\mathbf{0}_n \leq \overline{e}_k \leq \varepsilon_k$ and $\mathbf{0}_n \leq \underline{e}_k \leq \varepsilon_k$. So, any observer gain $L$ that stabilizes the open loop system \eqref{eq:farmer_error}, should stabilize \eqref{eq:err_sys}, and \eqref{eq:compare_sys}.      
\end{proof}
\subsection{Control Design}
Based on Lemma \eqref{lem:separation}, the design of the observer and control gains can be done separately, where $L$ can be designed first to stabilize \eqref{eq:farmer_error}. In our previous work~\cite{9790824} (also summarized in Proposition \ref{prop:frame}) we provided linear matrix inequalities, through which such an $L$ can be computed. 
%multiple design approaches to obtain an observer gain that stabilizes \eqref{} by solvig semidefinite, mixed-integer linear, and mixed-integer semidefinte programs, respectively. 
Given this observer gain, in order to obtain tractable LMIs to synthesize stabilizing dynamic control gains $A_c,B_c,C_c,D_c,K^u_{\nu},K^x_{\nu}$, the main challenge to be addressed is to resolve the bilinearities between decision variables when applying the existing results on the stability of Lipschitz nonlinear DT systems to \eqref{eq:compare_sys}. The following theorem tackles this challenge by applying similarity transformations and a change of variables.    
\begin{thm}\label{thm:control_design}
Suppose $L \in \mathbb{R}^{n \times l}$ is computed through Proposition \ref{prop:frame}, $\alpha > 0$ is a chosen real constant (e.g., is a desired decay rate for the error system \eqref{eq:err_sys}) picked by the designer, and $\epsilon=\frac{1}{\alpha \gamma^2}-1$, where $\gamma=\|F_{\phi}\|_{\infty}$ is the Lipshitz constant of the decomposition function $\phi_d$. Let
$(\mu_*,\Gamma_*,Q_*,\Theta_*)$ be a solution to the following SDP:

\vspace{-.3cm}
{\small
\begin{align}\label{eq:stabilizing_K}
&\min\limits_{\{\mu > 0, \Gamma \succ \mathbf{0}_{5n, 5n}, Q \succ \mathbf{0}_{5n,5n}, \Theta\}}{\mu}\\
\nonumber &\text{s.t.} \ \begin{bmatrix} \Gamma - Q & Q & Q\hat{A}^\top +\Theta \hat{B}^\top & I_{5n} \\
* & -\alpha I_{5n} & \mathbf{0}_{5n, 5n} & \mathbf{0}_{5n,5n} \\
* & * & -\frac{1}{2}Q & Q \\
* & * & * & Q-2\epsilon \Gamma \end{bmatrix} \prec \mathbf{0}_{20n,20n},\\
\nonumber & \quad \ \ \begin{bmatrix}
-\mu I_{\tilde{n}} & \Lambda^\top  & \Lambda^\top  \\
* & -\frac{1}{2}Q & \mathbf{0}_{5n,5n} \\
* & * & -\Gamma
\end{bmatrix} \prec \mathbf{0}_{\hat{n}, \hat{n}}, \ \begin{bmatrix} I_{4n} & Q \\
Q & \Gamma \end{bmatrix} \succeq \mathbf{0}_{10n, 10n}, 
\end{align}
}

\noindent where $\tilde{n}\triangleq 3(n_w+n_v)$ is the dimension of the augmented noise vector $\eta_k \in \mathbb{R}^{\tilde{n}}$ and $\hat{n} \triangleq \tilde{n}+10n$. Moreover, $\hat{A} \in \mathbb{R}^{5n \times 5n}$ and $\hat{B} \in \mathbb{R}^{5n \times 15n}$ are defined as follows:

\vspace{-.3cm}
{\small
\begin{align*}
%\tilde{A}_0 \triangleq \begin{bmatrix} 
[\hat{A}_{ij}]&=[\tilde{A}_{ij}],  i \in \{1,2\},  j \in \{1,2,3,4,5\},\\
[\hat{A}_{3j}]&=\mathbf{0}_{n,n}, \quad \quad \quad \quad \ \  j \in \{1,2,3,4,5\},\\
[\hat{A}_{41}]&= LC+|L|F_{\psi}, \quad \quad [\hat{A}_{42}]= |L|F_{\psi}, \quad [\hat{A}_{43}]=\mathbf{0}_{n,n},\\
[\hat{A}_{44}]&=(A-LC)^\oplus+LC, [\hat{A}_{45}] = -(A-LC)^\ominus, \\
[\hat{A}_{51}]&=-|L|F_{\psi}, \quad \quad \quad [\hat{A}_{52}]= LC-|L|F_{\psi}, [\hat{A}_{53}]=\mathbf{0}_{n,n},\\
[\hat{A}_{54}]&=-(A-LC)^\ominus, \quad \quad [\hat{A}_{55}]=(A-LC)^\oplus +LC,
\end{align*}
\begin{align*}
[\hat{B}_{ij}]\hspace{-.1cm}=\hspace{-.1cm}
\begin{cases} \mathbf{0}_{n,4n+l+3m},  \quad \quad \quad \quad i \in \{1,2\},  j \in \{1,2,3,4,5\},\\
 \mathbf{0}_{n,4n+l+3m},  \quad \quad \quad \quad   i \in \{3,4\},  j \in \{1,2\},\\
 \begin{bmatrix}I_n & \mathbf{0}_{n,3n+l+3m}\end{bmatrix},  \quad \  i=3, j=3,\\
 \begin{bmatrix}\mathbf{0}_{n,n} & I_n & \mathbf{0}_{n,n+2m+l} & I_n & \mathbf{0}_{n,m}\end{bmatrix}, \quad  i =3, j=4,\\
 \begin{bmatrix}\mathbf{0}_{n,2n} & -I_n & \mathbf{0}_{n,2m+l} & -I_n & \mathbf{0}_{n,m}\end{bmatrix}, \  i =3, j=5,\\
 \begin{bmatrix} \mathbf{0}_{n,4n+2m+l} & B\end{bmatrix}, \quad \ \  i \in \{4,5\}, j\in \{1,2\},\\
  \begin{bmatrix} \mathbf{0}_{n,3n} & B & \mathbf{0}_{n,n+3m}\end{bmatrix},
   \ i \in \{4,5\}, j=3,\\
 \begin{bmatrix} \mathbf{0}_{n,3n+l} & B & \mathbf{0}_{n,n+2m}\end{bmatrix}\hspace{-.1cm},  i \in \{4,5\}, j=4,\\
  \begin{bmatrix} \mathbf{0}_{n,3n+l+m} & -B & \mathbf{0}_{n,n+m}\end{bmatrix}, \quad \ i \in \{4,5\}, j=5.
\end{cases}
\end{align*}
}

%Moreover, . Finally, $\epsilon=\frac{1}{\alpha \gamma^2}-1$, where $\gamma=\|F_{\phi}\|_{\infty}$. 
Then, the closed-loop system \eqref{eq:compare_sys} is ISS and satisfies 
\begin{align}\label{eq:noise-attenuation}
\|z_k\|^2_2 \leq \mu_*\|\tilde{w}_k\|^2_2, \ \forall k,
\end{align}
with the dynamic control gains $A^*_c,B^*_c=[-\underline{K}^*_b \ \overline{K}^*_b], C^*_c, D^*_c=[-\underline{K}^*_d \ \overline{K}^*_d], {K}^{x*}_{\nu}$ and ${K}^{u*}_{\nu}$ that are obtained as follows: $\tdiag_5(\tilde{K}^*)
=(Q_*^{-1}\Theta_*)^\top$, where %block diagonal ``augmented control matrix":
{\small
\begin{align}\label{eq:control_design}
%\begin{array}{c}
%\tdiag_5(\tilde{K}^*)
%=(Q_*^{-1}\Theta_*)^\top, \ \text{where} \\
\hspace{-.4cm}\tilde{K}^*\hspace{-.1cm}\triangleq \hspace{-.1cm}\begin{bmatrix}A^{*\top}_c & \hspace{-.07cm} \overline{K}^{*\top}_b &\hspace{-.07cm} \underline{K}^{*\top}_b &\hspace{-.07cm} C^{*\top}_c & \hspace{-.07cm}\overline{K}^{*\top}_d &\hspace{-.07cm}  \underline{K}^{*\top}_d &\hspace{-.07cm} (\tilde K^{x*}_{\nu})^\top &\hspace{-.07cm} (\tilde K^{u*}_{\nu})^\top\hspace{-.07cm}\end{bmatrix}^\top \hspace{-.2cm}.
%\nonumber&\tilde K^{x*}_{\nu} \triangleq {K}^{x*}_{\nu}F_{\phi}, \tilde K^{u*}_{\nu} \triangleq {K}^{u*}_{\nu}F_{\phi}
%\end{array}
\end{align}
}

%\vspace{-1.3cm}
\noindent Moreover, $\tilde K^{x*}_{\nu} \triangleq {K}^{x*}_{\nu}F_{\phi}$, and $\tilde K^{u*}_{\nu} \triangleq {K}^{u*}_{\nu}F_{\phi}$.
%where $\tilde{K}^*\triangleq [A^{*\top}_c  \overline{K}^{*\top}_b  \underline{K}^{*\top}_b  C^{*\top}_c  \overline{K}^{*\top}_d  \underline{K}^{*\top}_d ({K}^{x*}_{\nu}F_{\phi})^\top ({K}^{u*}_{\nu}F_{\phi})^\top]^\top$
%$\overline{K}_* =$, $\underline{K}_* =$ and ${K}_{n*} =$. 
%Furthermore,
\end{thm}
%\vspace{-.2cm}
\begin{proof}
First, note that the augmented system in \eqref{eq:compare_sys} has a Lipschitz nonlinear component $\lambda$. This is because the mapping $\phi$ has bounded Jacobians by Assumption \ref{ass:mixed_monotonicity} and \eqref{eq:JSS_Jacobian}, and so is Lipschitz. Hence, $\phi_d$ is also Lipschitz by construction (cf. \eqref{eq:JJ_decomp} and \eqref{eq:Dj} in Proposition \ref{prop:tight_decomp}), with the Lipschitz constant $\gamma=\|F_{\phi}\|_{\infty}$. Then, by \cite[Lemma 3]{abbaszadeh2009lmi}, \eqref{eq:compare_sys} is ISS and satisfies \eqref{eq:noise-attenuation} if the following LMIs hold (which are the special cases of \cite[(15) and (16)]{abbaszadeh2009lmi} with $H$ being identity, while $N$ and $M$ are set to be zero matrices):

\vspace{-.3cm}
{\small
\begin{align}\label{eq:LMI_1}
&\begin{bmatrix}
I_{5n}-P & I_{5n} & \tilde{A}^\top P & \mathbf{0}_{5n,5n} \\
* & -\alpha I_{5n} & \mathbf{0}_{5n,5n} & \mathbf{0}_{5n,5n} \\
* & * & -\frac{1}{2}P & P \\
* & * & * & P-2\epsilon I_{5n}
\end{bmatrix} \prec \mathbf{0}_{20n,20n}, \\
\label{eq:LMI_2}&\begin{bmatrix}
-\mu I_{\tilde{n}} & \Lambda^\top P & \Lambda^\top P \\
* & -\frac{1}{2}P & \mathbf{0}_{5n,5n} \\
* & * & -I_{5n}
\end{bmatrix} \prec \mathbf{0}_{\hat{n},\hat{n}}.
\end{align}
}

\noindent \eqref{eq:LMI_1} and \eqref{eq:LMI_2} cannot be directly used to design the control gains due to the existence of the terms $P\tilde{A}$ which contains bilinearities in the form of pre- and post- multiplication of matrix $B$ with decision variables, i.e., the control gains. To overcome this difficulty, we apply two similarity transformations by pre- and post-multiplying the LMIs in \eqref{eq:LMI_1} and \eqref{eq:LMI_2} by $\text{diag}(P^{-1},I_{4n},P^{-1},P^{-1})$ and $\text{diag}(I_{\tilde{n}},P^{-1},P^{-1})$, respectively. This, combined with applying the change of variables $Q=P^{-1}$, and defining the new matrix variable $\Gamma$ that  satisfies $\Gamma \succeq P^{-2}=Q^2 \Leftrightarrow \begin{bmatrix} I_{4n} & Q \\
Q & \Gamma \end{bmatrix} \succeq \mathbf{0}_{8n \times 8n}$ by Schur complements, returns the results.  
\end{proof}
\section{Illustrative Example}\label{sec:example}
%The effectiveness of our control design is illustrated by the following example.
Consider a slightly modified version of the system in \cite[Example 2]{abbaszadeh2009lmi} in the form of \eqref{eq:system}--\eqref{eq:JSS_decom} with:

\vspace{-.3cm}
{\small
\begin{align*}
\label{eq:exampletwo}
\begin{array}{rl}
A&=\begin{bmatrix}
   0.5000 & -0.5975 & 0.3735 & 0.0457 & 0.3575\\
   0.2500 & 0.3000 & 0.4017  & 0.1114 & 0.0227\\
   0.4880 & 01384  & 0.2500  & 0.7500 & 0.7500\\
   0.3838 & 0.0974 & 0.5000  & 0.2500 & 0.5000\\
   0.0347 & 0.1865 & -0.2500 & 0.5000 & 0.2500
   \end{bmatrix},\\
   \phi(x)&=\begin{bmatrix}
   0.1(\sin(x_3)-x_3)\\
   0.2(\sin(x_4)-x_4)\\
   0.3(\sin(x_1)-x_1)\\
   0\\
   0.1(\sin(x_2)-x_2)
   \end{bmatrix},
   B=\begin{bmatrix}
   0.7 & 0.8 & 0\\
   0.4 & 0.9 & 0.9\\
   0.9 & 0.9 & 0.2\\
   0.9 & 0.6 & 0.7\\
   0   & 0.5 & 0.3
   \end{bmatrix},\\
   D&=\begin{bmatrix}
   0.2 & 0.1 & 0\\
   0.2 & 0.1 & 0
   \end{bmatrix},
   C=\begin{bmatrix}
   0.5 & 0.2 & 0 & 0 & 0.3\\
   0   & 0.2 & 0.1 & 0.3 & 0
   \end{bmatrix}, V=I_2\\
   \psi(x)&=[0.1(\cos(x_1)-x_1) \ 0.2(\cos(x_2)+x_2)]^\top, W=I_5,
\end{array}
\end{align*}
}

\noindent where %$n=2$, %$ y_t \in \mathbb{R} $, 
$\mathcal{W}=[-.1 \ .1]^5,\mathcal{V}=[-.1 \ .1]^2$ and $\mathcal{X}_0 = [-6 \ 6]^5$. The nonlinear system is unstable. Moreover, matrix $A$ is unstable, which means the linear part of the system is itself unstable. The first plot in Figure \ref{fig:states} (i.e., in the upper left corner) shows the trajectory of all the states for the unstable open-loop system, where some of the state values grow unbounded as expected. All the computed system parameters and observer gains are given in Figure \ref{fig:wideMatrix}. 
\begin{figure*}[!t]  % Use [!t] for top, [!b] for bottom
\centering
{\scriptsize%\footnotesize
\begin{align*}
\begin{array}{rl}
A_c\hspace{-.1cm}&=\hspace{-.1cm}\begin{bmatrix}
   0.0400 & -0.0864 & 0.2624 & 0.0046 & 0.2464\\
   0.1400 & 0.2000 & 0.3006  & 0.0003 & 0.0116\\
   0.3779 & 0.0273  & 0.1499  & 0.6499 & 0.6499\\
   0.2727 & 0.9863 & 0.4999  & 0.1499 & 0.4198\\
   0.9236 & 0.0754 & -0.1499 & 0.4999 & 0.1497
   \end{bmatrix}, \overline{K}_b \hspace{-.1cm}=\hspace{-.1cm}-\underline{K}_b\hspace{-.1cm}=\hspace{-.1cm}\begin{bmatrix} 
   -0.6575 & 0.5494 & 0 & 0 & 0 \\ 
-0.2997 & -0.1977 & 0 & 0 & 0.1 \\ 
0.2196 & -0.5392 & 0.2 & 0 & 0\\
0 & 0 & 0 & 0.1 & 0 \\ 
0 & 0 & 0 & 0 & 0.15
\end{bmatrix}, B_c\hspace{-.1cm}=\hspace{-.1cm}\begin{bmatrix}
   -0.5 & 0.6 & 0\\
   0.213 & -0.7001 & 0.72\\
   -0.7001 & 0.8502 & -0.1008\\
   0.6481 & -0.4 & 0.5719\\
   0   & 0.3002 & -0.1
   \end{bmatrix},\\
   {C}_{c} \hspace{-.1cm}&=\hspace{-.1cm}\begin{bmatrix} 0.2232 & -0.2191 & 0.1 & 0 & -0.01\\ 
0.0993 & -0.0933 & 0 & 0 & -0.04\\ 
-0.0002 & 0.21 & 0 & -0.1002 & 0\end{bmatrix},
\overline{K}_{d}\hspace{-.1cm}=\hspace{-.1cm}-\underline{K}_{d} \hspace{-.1cm}=\hspace{-.2cm}\begin{bmatrix} 0.1121 & -0.108 & 0 & 0 & -0.9099\\ 
0.0882 & -0.0812 & 0 & 0 & -0.03\\ 
-0.0001 & 0.1 & -0.1 & 0.1002 & 0\end{bmatrix}, F_{\phi}=\begin{bmatrix} \varepsilon_0 & 0 & 0.2 & 0 & 0  \\ 0 & \varepsilon_0 & 0.4 & 0 & 0\\0.6 &  0 & \varepsilon_0 & 0 & 0 \\ 0 & 0 & 0 & \varepsilon_0 & 0\\ 0 & 0.2 & 0 & 0 & \varepsilon_0  \end{bmatrix}, \\ {K}^x_{\nu}\hspace{-.1cm}&=\hspace{-.1cm}\begin{bmatrix} 
   -0.5464 & 0.4383 & 0 & 0 & 0 \\ 
-0.1886 & -0.0866 & 0 & 0 & 0.1 \\ 
0.1085 & -0.4281 & 0 & 0 & 0\\
0 & 0 & 0.5392 & 0 & 0 \\ 
0 & 0 & 0 & -0.3403 & 0
\end{bmatrix}\hspace{-.1cm}, {K}^u_{\nu} \hspace{-.1cm}=\hspace{-.1cm}\begin{bmatrix} 0.0010 & -0.0979 & 0 & 0 & -0.8088\\ 
0.0771 & -0.0701 & 0 & 0 & -0.0200\\ 
-0.1001 & 0.2000 & -0.1 & 0 & 0\end{bmatrix}\hspace{-.1cm}, F_{\psi}\hspace{-.1cm}=\hspace{-.1cm}\begin{bmatrix} 0.2 & 0 \\ 0 & 0.4 \end{bmatrix}\hspace{-.1cm}, L\hspace{-.1cm}=\hspace{-.1cm} \begin{bmatrix}        0.1   &  0   &   0 & 0 & 0\\     -20.0999 &  1    & 0 & 0 & 0      \end{bmatrix}\hspace{-.1cm}.
\end{array}
\end{align*}
}
\vspace{-.2cm}
\caption{{\small Parameters and computed gains for the example system in Section \ref{sec:example}.}}
\label{fig:wideMatrix}
\end{figure*}
Using Proposition \ref{prop:tight_decomp} we computed $F_{\phi}$ and $F_{\psi},\gamma=\|F_{\phi}\|_{\infty}=0.6010$, where $\varepsilon_0=0.001$ is a small number added to the diagonal entries of the original $F_{\phi}$ to make it an invertible matrix. We set $\alpha=0.1$ which results in $\epsilon =\frac{1}{\alpha \gamma^2}-1=26.6854$. Moreover, based on the results of the separation principle in Lemma \ref{lem:separation}, we first obtained $L=$ by solving the SDP in \cite[(17) and (19)]{9790824}. Then, given $L$, solving the SDP in \eqref{eq:stabilizing_K} returned the control gains,
as well as the optimal noise attenuation level $\mu_*=1.2026$. 
%was obtained. 
The upper and lower closed-loop framers, i.e., $\overline{x}^{dy},\underline{x}^{dy}$ are shown in
the second to sixth plots in Figure \ref{fig:states}, where the corresponding framers after applying a static feedback control strategy \cite{khajenejad2024optimalcontrolstatic}, i.e., $\overline{x}^{dy},\underline{x}^{dy}$ are also plotted. As can be seen, both strategies are successful in stabilizing the trajectories; however, the dynamic control strategy significantly outperforms the static approach resulting in much tighter interval estimates.  %\mk{($x_1$ omitted for brevity)}

\begin{figure}[t!] % for sub figures over two columns in 
\centering
%{\includegraphics[width=0.49\columnwidth,trim=10mm 15mm 10mm 0mm]{Figures/F1_dis.eps}}\label{fig:sub4}
%\hfil
%{\includegraphics[width=0.49\columnwidth,trim=10mm 15mm 10mm 0mm]{Figures/F2_dis.eps}}\label{fig:sub5}
%\hfil
{\includegraphics[width=0.48\columnwidth]{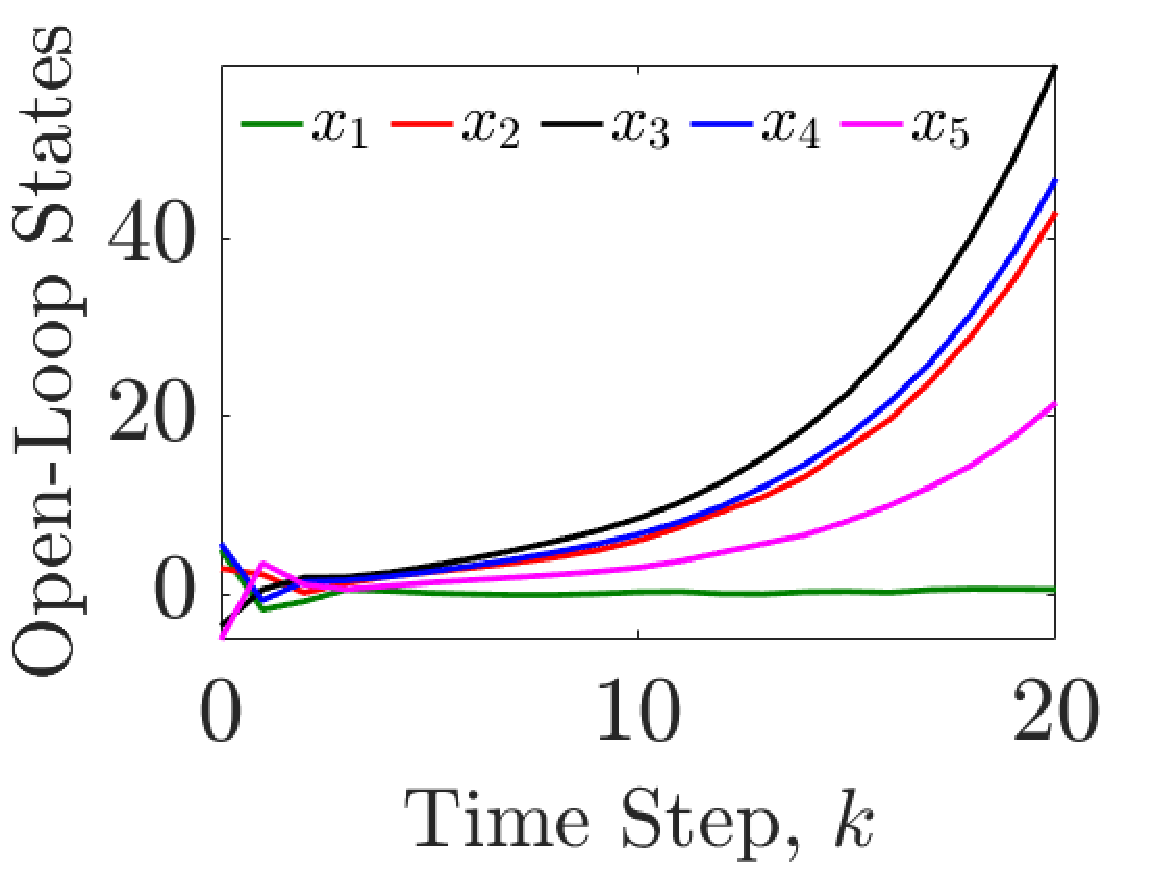}}
{\includegraphics[width=0.48\columnwidth]{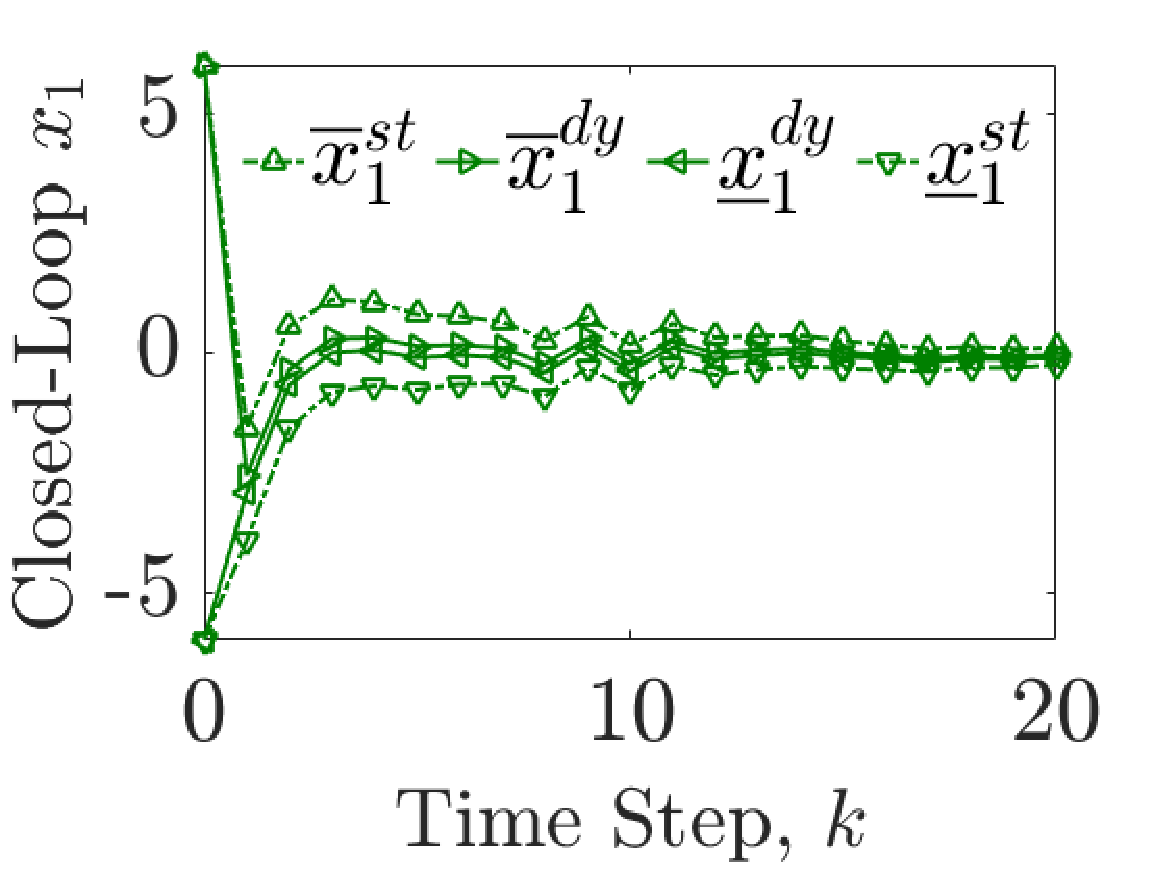}}
{\includegraphics[width=0.48\columnwidth]{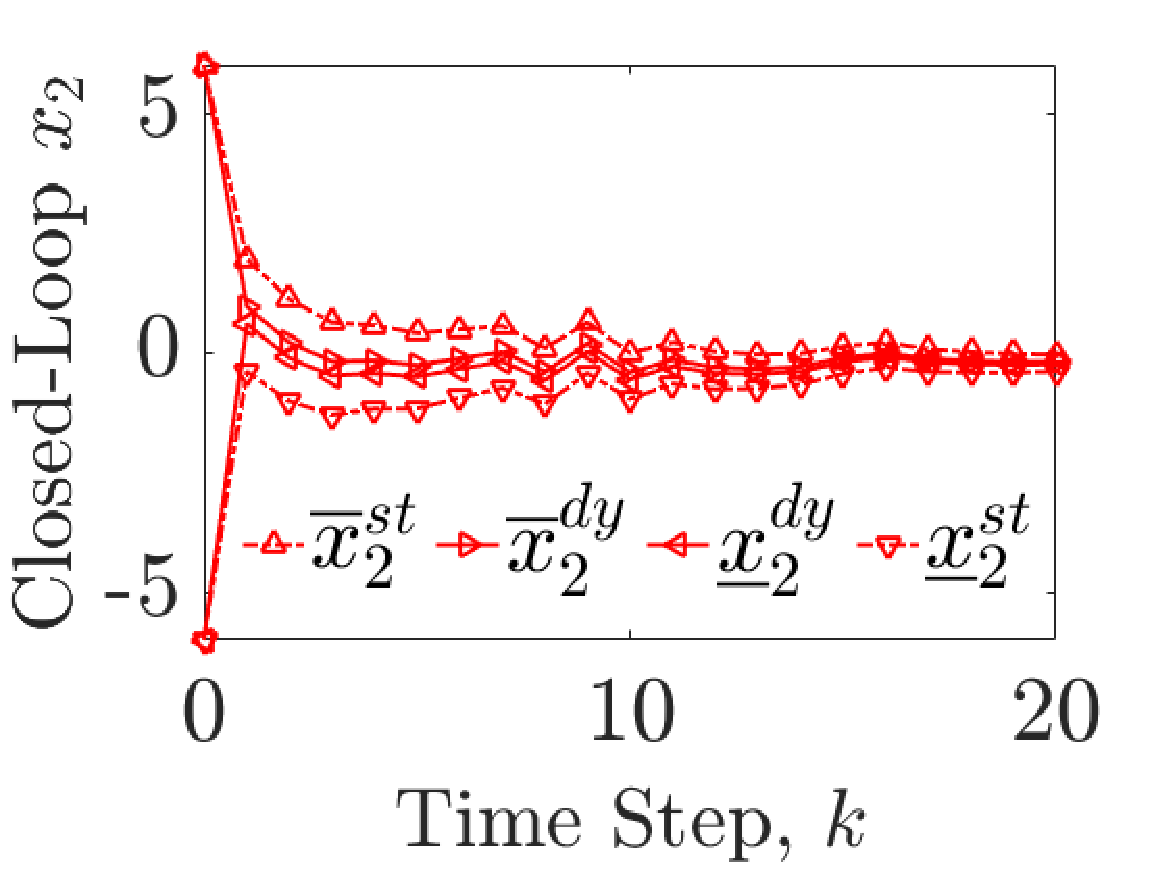}}
{\includegraphics[width=0.48\columnwidth]{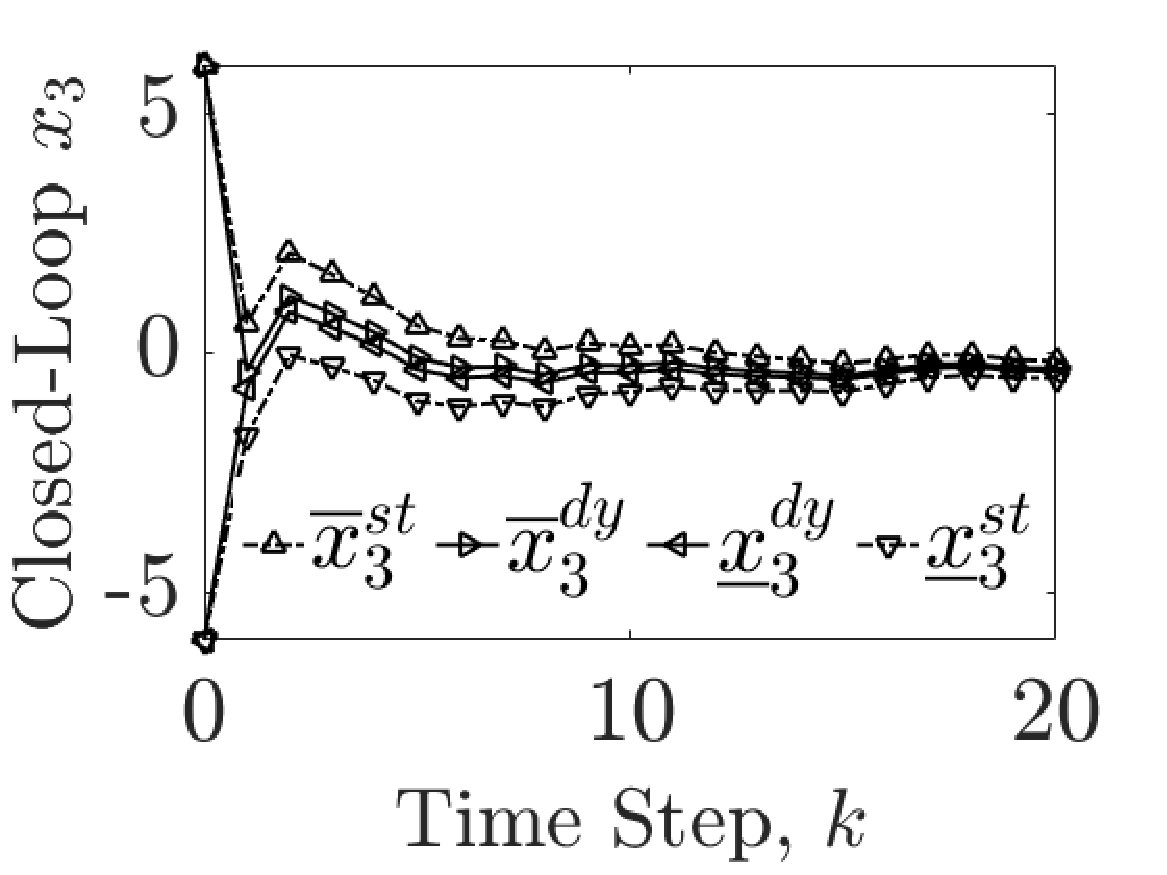}}
{\includegraphics[width=0.48\columnwidth]{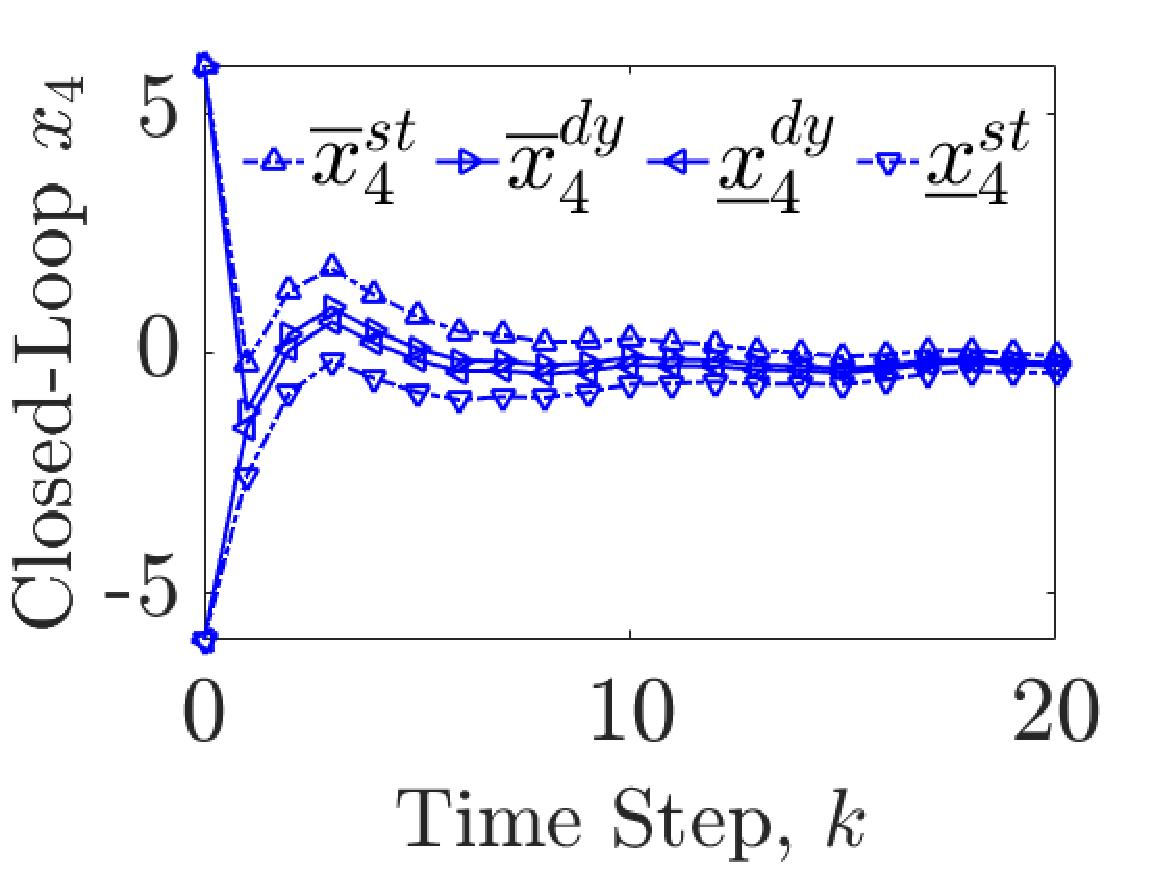}}
{\includegraphics[width=0.48\columnwidth]{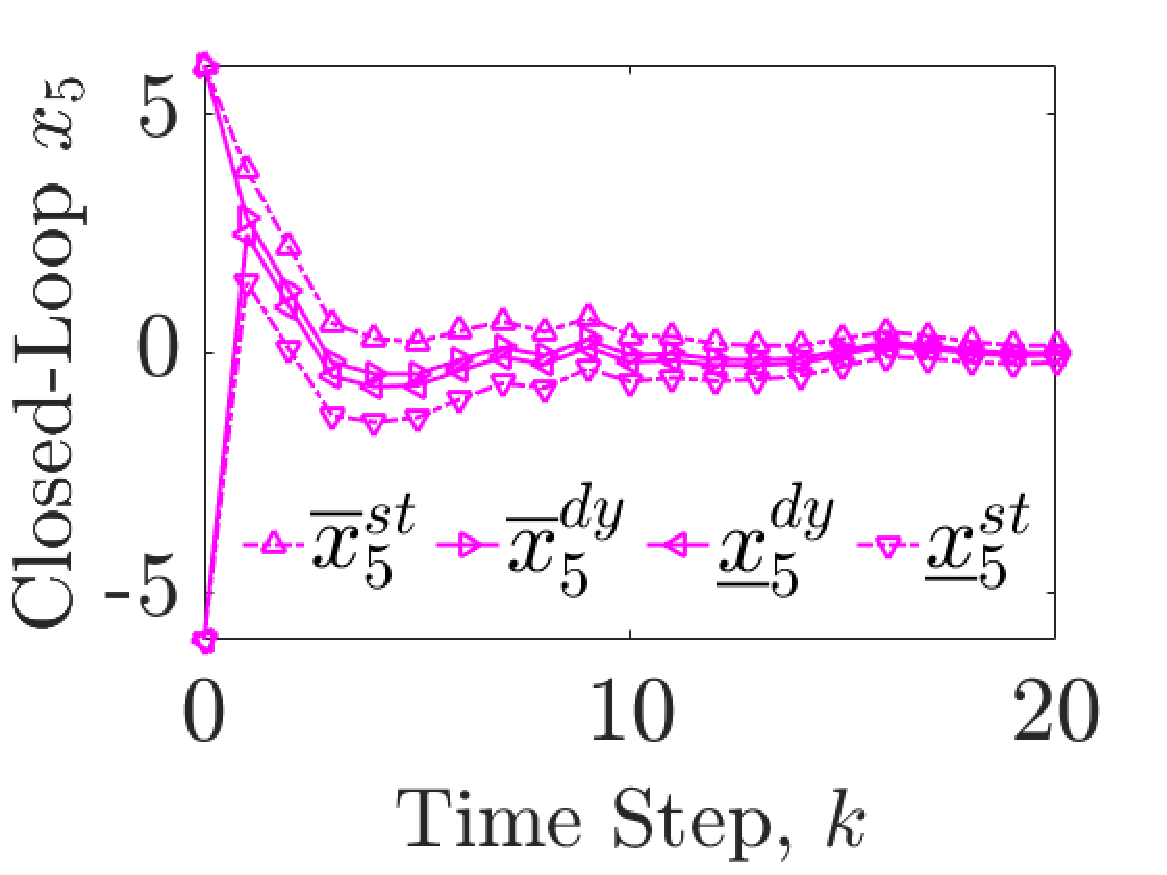}}
%{\includegraphics[scale=.155,trim=55mm 0mm 0mm 0mm]{Figures/states_6.png}}
%{\includegraphics[width=\columnwidth,trim=10mm 15mm 10mm 0mm]{Figures/F1_new.eps}}\label{fig:sub4}
%{\includegraphics[scale=.4,trim=10mm 15mm 10mm 0mm]{Figures/F1_new.eps}}\label{fig:sub4}
%\hfil
%{\includegraphics[width=0.48\columnwidth,trim=7mm 5mm 10mm 0mm]{Figures/error_Linf_DT.eps}}\label{fig:sub5}
%{\includegraphics[width=\columnwidth,trim=10mm 15mm 10mm 0mm]{Figures/F2_new.eps}}\label{fig:sub5}
%{\includegraphics[scale=.4,trim=10mm 15mm 10mm 0mm]{Figures/F2_new.eps}}\label{fig:sub5}
%\hfil
\caption{\small\mk{Open-loop states (first plot) and the closed-loop upper and lower framers returned by our proposed dynamic control design, i.e., $\overline{x}^{dy},\underline{x}^{dy}$, as well as the framers returned by the static feedback control approach in \cite{khajenejad2024optimalcontrolstatic}, i.e., $\overline{x}^{st},\underline{x}^{st}$ (second to sixth plots).}}%{States, $x_1,x_2$, as well as upper and lower framers, \mk{returned by our proposed observer} $\overline{x}_1,\overline{x}_2,\underline{x}_1,\underline{x}_2$, \mk{and by the observer in \cite{tahir2021synthesis}, $\overline{x}^{TA}_2,\overline{x}^{TA}_3,\underline{x}^{TA}_2,\underline{x}^{TA}_3$} for the DT System example.}}
\label{fig:states}
\end{figure}
%\vspace{-.4cm}
\section{Conclusion and Future Work} \label{sec:conclusion}
In this paper, an optimal dynamic stabilizing control for bounded Jacobian nonlinear discrete-time systems with nonlinear observations, subject to state and process noise, was presented. Rather than stabilizing the uncertain system directly, a higher-dimensional interval observer was stabilized, with its states containing the actual system states. The nonlinear control approach provided additional flexibility compared to linear and static methods, compensating for system nonlinearities and allowing tighter closed-loop intervals. A separation principle was established, enabling the design of observer and control gains, and tractable programs were derived for control gain design. An Extension to switched and hybrid systems, as well as other control strategies such as model predictive control will be considered for future work.
\balance
\bibliographystyle{unsrturl}

\bibliography{biblio}

\end{document}